\documentclass[12pt,reqno]{amsart}

\usepackage{times}
\usepackage{graphicx}
\graphicspath{ {./images/} }

\usepackage{amsmath,amsthm,amsfonts,amscd,amssymb}
\usepackage{pst-node}
\usepackage{pst-plot}
\usepackage[all]{xy}

\usepackage{tikz, braids}
\usetikzlibrary{decorations.pathreplacing, arrows}

\newtheorem{theorem}{Theorem}[section]

\newtheorem{definition}[theorem]{Definition}
\newtheorem{lemma}[theorem]{Lemma}
\newtheorem{remark}[theorem]{Remark}
\newtheorem{proposition}[theorem]{Proposition}

\numberwithin{theorem}{section}
\numberwithin{equation}{section}

\newcommand{\mc}{\mathcal}

\newcommand{\la}{\langle}
\newcommand{\ra}{\rangle}

\newcommand{\Comp}{\mathbb{C}}

\newcommand{\n}{\mathbb{N}}
\begin{document}

\title{Irreducibly $SU(2)$-covariant quantum channels of low rank}
% or if you want, simply \title{Title of the article}

\author{Euijung Chang}
\address{Euijung Chang, 
Department of Mathematics Education, Seoul National University, 
Gwanak-Ro 1, Gwanak-Gu, Seoul 08826, Republic of Korea}
\email{uijeong97@snu.ac.kr }

\author{Jaeyoung Kim}
\address{Jaeyoung Kim, 
Department of Mathematics Education, Seoul National University, 
Gwanak-Ro 1, Gwanak-Gu, Seoul 08826, Republic of Korea}
\email{young15360@snu.ac.kr }

\author{Hyesun Kwak}
\address{Hyesun Kwak, 
Department of Mathematics Education, Seoul National University, 
Gwanak-Ro 1, Gwanak-Gu, Seoul 08826, Republic of Korea}
\email{hskwak@snu.ac.kr }

\author {Hun Hee Lee}
\address{Hun Hee Lee,
Department of Mathematical Sciences and the Research Institute of Mathematics, Seoul National University, Gwanak-ro 1, Gwanak-gu, Seoul 08826, Republic of Korea}
\email{hunheelee@snu.ac.kr}

\author{Sang-Gyun Youn}
\address{Sang-Gyun Youn, 
Department of Mathematics Education, Seoul National University, 
Gwanak-Ro 1, Gwanak-Gu, Seoul 08826, Republic of Korea}
\email{s.youn@snu.ac.kr }

\maketitle

\begin{abstract}
We investigate information theoretic properties of low rank (less than or equal to 3) quantum channels with $SU(2)$-symmetry, where we have a complete description. We prove that PPT property coincides with entanglement-breaking property and that degradability seldomly holds in this class. In connection with these results we will demonstrate how we can compute Holevo and coherent information of those channels. In particular, we exhibit a strong form of additivity violation of coherent information, which resembles the superactivation of coherent information of depolarizing channels.
\end{abstract}

\markleft{\footnotesize Euijung Chang, Jaeyoung Kim, Hyesun Kwak, Hun Hee Lee, Sang-Gyun Youn}

\section{Introduction}\label{sec-intro}

The concept of symmetry is ubiquitous in quantum theory, which is usually described by groups through their representation theory. Among those group symmetry the one by $SU(2)$ is widely-known thanks to its connection to non-relativistic spins. In this article we would like to focus on $SU(2)$-symmetry for quantum channels. Let us begin with a quantum channel $\Phi: B(H_A) \to B(H_B)$, where dim$H_A=k \in \n$ and dim$H_B=l \in \n$ are finite dimensional Hilbert spaces. It is well known that the group $SU(2)$ has a complete family of irreducible unitary representations $\{\pi_m: m\ge 0\}$, where dim$\pi_m = m+1$. This means that for each dimension $k\ge 1$ we have a unique (upto unitary equivalence) choice of irreducible unitary representation $\pi_{k-1} : SU(2) \to U(k)$, where $U(k)$ is the unitary group of degree $k$. This allows us to consider {\it (irreducible) $SU(2)$-covariance} of the channel $\Phi$, namely the condition \small
    $$\Phi(\pi_{k-1}(x)X \pi_{k-1}(x)^*) = \pi_{l-1}(x)\Phi(X)\pi_{l-1}(x)^*,\;\; \forall x\in SU(2), \; \forall X \in M_k(\Comp).$$
\normalsize In this case we say that $\Phi$ is {\it $SU(2)_{(k-1,l-1)}$-covariant} or {\it irreducibly $SU(2)$-covariant}. A detailed study of irreducibly $SU(2)$-covariant channels has been initiated by \cite{Al14}, which has been generalized to the case of compact quantum group symmetry in \cite{LeYo20}. One of the main result in \cite{Al14} says that the set of all $SU(2)_{(k,l)}$-covariant quantum channels is a simplex whose extreme points can be effectively listed via $SU(2)$-representation theory (see Section \ref{sec-SU(2)-cov-ch} for the details). In other words, we have a tractable parameterization for irreducibly $SU(2)$-covariant channels.

The above is the starting point of the current paper, which investigates information theoretic properties of irreducibly $SU(2)$-covariant channels based on the given parameterization. Recall that we can measure how much information can be transmitted through a given quantum channel via channel capacities. Among many of those the {\it classical capacity $C(\Phi)$} and the {\it quantum capacity $Q(\Phi)$} are of the most importance. However, computing either of the capacities for a given quantum channel is a highly complicated task, partly because they are expressed as {\it ``regularizations"} of simpler quantities. More precisely, we have
    \[C(\Phi)=\lim_{n\rightarrow \infty}\frac{\chi(\Phi^{\otimes n})}{n}\;\;\text{and}\;\; Q(\Phi)=\lim_{n\rightarrow \infty}\frac{Q^{(1)}(\Phi^{\otimes n})}{n},\]
where $\chi(\Phi)$ and $Q^{(1)}(\Phi)$ are the {\it Holevo information} and the {\it coherent information} of $\Phi$. For this reason, it is also important to figure out when we could possibly avoid ``regularizations" for computing channel capacities. Indeed, we have (\cite{Sh02}) $C(\Phi)=\chi(\Phi)$ if $\Phi$ is {\it entanglement-breaking} (shortly, EBT) and (\cite{DeSh05}) $Q(\Phi)=Q^{(1)}(\Phi)$ if $\Phi$ is {\it degradable}. In this paper, we, first, would like to provide characterizations of EBT property and degradability for irreducibly $SU(2)$-covariant channels in terms of their parameters given by the structure theorem. We will, however, restrict ourselves on low rank cases, namely the cases of irreducibly $SU(2)$-covariant channels from (resp. into) $B(\Comp^2)$ or from $B(\Comp^3)$ for simplicity. In the mentioned cases the channels have rank $\le 3$ as linear maps and the resulting simplices are a 2-simplex (or a line segment) and a 3-simplex (or a triangle), where we have 2 and 3 real parameters, respectively. We, then, will move our attention back to the quantities $\chi(\Phi)$ and $Q^{(1)}(\Phi)$ for the case of a single qubit or qutrit input. The former quantity turns out to be fully computable for the qubit and qutrit cases, which boils down to finding minimizers for output entropies. In particular, ``{\it the coherent states}'' are the optimizers for $\chi(\Phi)$ for qubit case, whereas it is not always true for the qutrit case. The latter quantity $Q^{(1)}(\Phi)$ also turns out to be relatively easy to compute for the qubit case, which allows us to observe an ``{\it almost superactivation}", namely
\[\frac{Q^{(1)}(\Phi\otimes \Phi)}{2}\geq 0.0039>Q^{(1)}(\Phi)\approx 0\]
for a choice of irreducibly $SU(2)$-covariant channel $\Phi$. Note that this can be understood as a strong form of additivity violation of coherent information. Here, $Q^{(1)}(\Phi)\approx 0$ is up to an error estimate $\epsilon \leq  10^{-6}$.

\section{The structure of irreducibly $SU(2)$-covariant maps and quantum channels}\label{sec-SU(2)-cov-ch}

\subsection{Representation theory of $SU(2)$ and $SU(2)$-covariance}

For a compact group $G$, a continuous function $\pi:G\rightarrow \mathcal{U}(H)$ is called a (finite dimensional) {\it unitary representation} if $\pi(xy)=\pi(x)\pi(y)$ for all $x,y\in G$. Here, $\mathcal{U}(H)$ is the set of all unitaries on a finite dimensional Hilbert space $H$. A subspace $V\subseteq H$ is called {\it invariant} if $\pi(x)V\subset V$ for all $x\in G$, and $\pi$ is called {\it irreducible} if $\left \{0\right\}$ and $H$ are the only invariant subspaces.

Recall the collection $\{\pi_k: k\ge 0\}$ of irreducible unitary representations of $SU(2)$, where $\pi_k :  SU(2) \to \mc U(H_k)$ with dim$H_k = k+1$, mentioned in the introduction. It is well-known that the above collection is complete in the sense that any irreducible unitary representation $\rho$ on $SU(2)$ is unitarily equivalent to $\pi_k$ for some $k\ge 0$. Meanwhile, the tensor product representation $\pi_l\otimes \pi_m : SU(2)\rightarrow \mathcal U(H_l \otimes H_m)$ is not irreducible if $l,m>0$, but it is unitarily equivalent to the direct sum $\pi_{|l-m|} \oplus \pi_{|l-m|+2} \oplus \cdots \oplus \pi_{l+m}$. This means that there is an onto unitary
    $$U^{l,m}:H_{|l-m|}\oplus H_{|l-m|+2}\oplus \cdots \oplus H_{l+m}\to H_l\otimes H_m$$
such that
    $$(U^{l,m})^* (\pi_l(x)\otimes \pi_m(x)) U^{l,m} = \pi_{|l-m|}(x) \oplus \cdots \oplus \pi_{l+m}(x),\;\; x\in SU(2).$$ 
When we compose the canonical embedding 
\[\gamma_k^{l,m}: H_k \hookrightarrow H_{|l-m|}\oplus \cdots \oplus H_{l+m}\] 
for $|l-m| \le k \le l+m$ with the unitary $U^{l,m}$ we get the isometry $\alpha^{l,m}_k = U^{l,m} \circ \gamma^{l,m}_k:H_k\to H_l\otimes H_m$ such that
    $$(\alpha^{l,m}_k)^* (\pi_l(x)\otimes \pi_m(x)) \alpha^{l,m}_k = \pi_k(x),\;\; x\in SU(2).$$
\begin{definition}
    The quantum channel $\Phi^{k\to l}_m: B(H_k) \to B(H_l)$ with Stinespring isometry $\alpha^{l,m}_k$, i.e.
        $$\Phi^{k\rightarrow l}_m(X):=(\text{id}\otimes \text{Tr})(\alpha^{l,m}_k X (\alpha^{l,m}_k)^*),\; X \in B(H_k),$$
    is called a {\bf $SU(2)$-Clebsch-Gordan} channel.
\end{definition}

The expression ``Clebsch-Gordan" in the above definition comes from the following identity
    \begin{equation}\label{eq-CG}
        \alpha^{l,m}_k |i\ra=\sum_{i_1=0}^l \sum_{i_2=0}^m C^{l,m,k}_{i_1,i_2,i}|i_1 i_2\ra.
    \end{equation}
Here, the constants $C^{l,m,k}_{i_1,i_2,i}$ are called the {\it Clebsch-Gordan coefficients} of $SU(2)$ and $|i\ra$ and $|i_1i_2\ra$ refer to the canonical basis for $H_k$ and $H_l \otimes H_m$.

It turns out that $SU(2)$-Clebsch-Gordan channels are building blocks for irreducibly $SU(2)$-covariant channels. We first record that the concept of covariance can be easily extended to the case of linear maps.

\begin{definition}
A linear map $\Phi: B(H_k) \to B(H_l)$ is called $SU(2)_{(k,l)}$-covariant if
    $$\Phi(\pi_k(x)X \pi_k(x)^*) = \pi_l(x)\Phi(X)\pi_l(x)^*,\;\; \forall x\in SU(2),\; \forall X\in B(H_k).$$
\end{definition}

The proof for the following can be found in \cite[Corollary 4.6, Proposition 5.1]{Al14} and \cite[Theorem 4.6]{LeYo20}.

\begin{theorem}
The collection $\left \{ \Phi^{k\rightarrow l}_{|k-l|},\Phi^{ k \rightarrow l}_{|k-l|+2},\cdots,\Phi^{k\rightarrow l}_{k+l}\right\}$ of $SU(2)$-Clebsch-Gordan channels is a basis of $\text{Cov}(k,l)$, the linear space of all $SU(2)_{(k,l)}$-covariant linear maps. Moreover, it is the set of all extreme points of the convex set $\text{CovQC}(k,l)$ of all $SU(2)_{(k,l)}$-covariant quantum channels. In particular, every $SU(2)$-Clebsch-Gordan channel is irreducibly $SU(2)$-covariant. 
\end{theorem}

\subsection{$SU(2)$-Clebsch-Gordan channels of low rank and their Kraus operators}

This paper focuses on the cases of irreducibly $SU(2)$-covariant quantum channels of low rank, more precisely irreducibly $SU(2)$-covariant channels from (resp. into) $B(\Comp^2)$ or on $B(\Comp^3)$. This means that we restrict our attention to the convex sets $\text{CovQC}(1,l)$, $\text{CovQC}(l,1)$ and $\text{CovQC}(2,2)$. Their Stinespring isometries $\alpha^{l,m}_k$ are given by \eqref{eq-CG}, and the explicit formulae of Clebsch-Gordan coefficients $C^{l,m,k}_{i_1,i_2,i}$ are given in \cite[Section V.2]{Bo01}.
%and the corresponding linear spaces $\text{QC}(1,l)$, $\text{QC}(l,1)$ and $\text{QC}(3,3)$. 

\subsubsection{$SU(2)_{(1,l)}$-covariant channels}\label{subsubsec-(1,l)}

The convex set $\text{CovQC}(1,l)$ is a line segment with two end points $\Phi^{1\rightarrow l}_{l-1}$ and $\Phi^{1\rightarrow l}_{l+1}$, which we give 
explicit descriptions below.

For the channel $\Phi^{1\rightarrow l}_{l-1}$ we have
\begin{align}\label{eq-l-1}
\Phi^{1\rightarrow l}_{l-1}(|0\ra\la 0|)& =\frac{2}{l(l+1)}\sum_{i=0}^l (l-i)\cdot |i\ra\la i|\\
\Phi^{1\rightarrow l}_{l-1}(|1\ra\la 1|)& =\frac{2}{l(l+1)}\sum_{i=0}^l i\cdot |i\ra\la i| \nonumber\\
\Phi^{1\rightarrow l}_{l-1}(|0\ra\la 1|)& = \frac{2}{l(l+1)}\sum_{i=0}^{l-1} \sqrt{(l-i)\cdot (i+1)}|i\ra\la i+1| \nonumber \\ &= (\Phi^{1\rightarrow l}_{l-1}(|1\ra\la 0|))^*,\nonumber
\end{align}
and $l$ Kraus operators of $\Phi^{1\rightarrow l}_{l-1}$ are given by
\begin{equation}\label{eq-Kraus1}
K_i = \sqrt{\frac{2}{l(l+1)}}\cdot (-1)^{i-1}\cdot \left (\sqrt{i} |l-i \ra \la 0 |+\sqrt{l-i+1}|l-i+1\ra\la 1 | \right )
\end{equation}
with $1\leq i\leq l$, i.e. $\displaystyle \Phi^{1\rightarrow l}_{l-1}(X) = \sum_{ i =1}^lK_i X K^*_i$ for any $X\in B(\Comp^2)$.

For the channel $\Phi^{1\rightarrow l}_{l+1}$ we have
\begin{align}\label{eq-l+1}
\Phi^{1\rightarrow l}_{l+1}(|0\ra\la 0|)&=\frac{2}{(l+1)(l+2)}\sum_{i=0}^l (i+1)\cdot |i\ra\la i| \\
\Phi^{1\rightarrow l}_{l+1}(|1\ra\la 1|)&=\frac{2}{(l+1)(l+2)}\sum_{i=0}^l (l-i+1)\cdot |i\ra\la i| \nonumber\\
\Phi^{1\rightarrow l}_{l+1}(|0\ra\la 1|)&=\frac{-2}{(l+1)(l+2)}\sum_{i=0}^{l-1} \sqrt{(l-i)\cdot (i+1)}|i\ra\la i+1| \nonumber\\ &= (\Phi^{1\rightarrow l}_{l+1}(|1\ra\la 0|))^*, \nonumber
\end{align}
and $l+2$ Kraus operators of $\Phi^{1\rightarrow l}_{l+1}$ are given by 
\begin{align}
\label{eq-Kraus2} K_{l+i}=\sqrt{\frac{2}{(l+1)(l+2)}} \cdot (-1)^{i-1} \cdot &\left ( \sqrt{l-i+2}  |l-i+1 \ra  \la 0 | \right .   \\
\notag & \color{white}ttttt \color{black}\left .  - \sqrt{i-1}|l-i+2\ra\la 1 | \right ) 
\end{align}
with $1\leq i\leq l+2$.

A general element of $\text{CovQC}(1,l)$ is of the form 
\begin{equation}\label{eq-Cov(1,l)-ch}
\Phi_p=(1-p)\Phi^{1\rightarrow l}_{l-1}+p\Phi^{1\rightarrow l}_{l+1}    
\end{equation}
with $0\leq p\leq 1$, whose Kraus operators are
    \[\sqrt{1-p}K_1,\cdots,\sqrt{1-p}K_l,\sqrt{p}K_{l+1},\cdots ,\sqrt{p}K_{2l+2}.\]

\subsubsection{$SU(2)_{(l,1)}$-covariant channels}\label{subsubsec-(l,1)}

The convex set $\text{CovQC}(l,1)$ is a line segment given by
\begin{equation}\label{formula-Cov(l,1)-ch}
    \Psi_p= (1-p)\Phi^{l\rightarrow 1}_{l-1}+p\Phi^{l\rightarrow 1}_{l+1}
\end{equation} 
with $0\leq p\leq 1$. The two end points $\Phi^{l\rightarrow 1}_{l-1}$ and $\Phi^{l\rightarrow 1}_{l+1}$ are obtained as the adjoint maps of $\displaystyle \frac{l+1}{2}\Phi^{1\rightarrow l}_{l-1}$ and $\displaystyle \frac{l+1}{2}\Phi^{1\rightarrow l}_{l+1}$, i.e. for any $X\in M_2(\Comp)$, $Y\in M_l(\Comp)$ and $k\in \left \{l-1,l+1\right\}$ we have
\[\text{Tr}\left ( \frac{l+1}{2}\Phi^{1\rightarrow l}_{k}(X)\cdot Y \right )=\text{Tr}(X\cdot \Phi^{l\rightarrow 1}_k(Y)).\]

Moreover, the Kraus operators of $\Phi^{l\rightarrow 1}_{l-1}$ and $\Phi^{l\rightarrow 1}_{l+1}$ are given by 
\begin{equation}\label{eq-Cov(l,1)-ch}
\sqrt{\frac{l+1}{2}} K_1^*,\cdots,\sqrt{\frac{l+1}{2}} K_l^*,\sqrt{\frac{l+1}{2}} K_{l+1}^*,\cdots,\sqrt{\frac{l+1}{2}} K_{2l+2}^*,
\end{equation}
which we denote by $L_1,L_2,\cdots,L_l,L_{l+1},\cdots,L_{2l+2}$ respectively.

\subsubsection{$SU(2)_{(2,2)}$-covariant channels}\label{subsubsec-(2,2)}

The convex set $\text{CovQC}(2,2)$ is a triangle with three vertices $\Phi^{2\rightarrow 2}_{0}$, $\Phi^{2\rightarrow 2}_{2}$ and $\Phi^{2\rightarrow 2}_4$, so a general element in $\text{CovQC}(2,2)$ is of the form 
\begin{equation}\label{eq-Cov(2,2)-ch}
\mathcal{N}_{p,q}=(1-p-q)\Phi^{2\rightarrow 2}_0+p\Phi^{2\rightarrow 2}_2+q\Phi^{2\rightarrow 2}_4
\end{equation}
for $0\leq p,q\leq 1$ with $p+q\leq 1$. For any $A=(a_{i,j})_{0\leq i,j\leq 2}$ the extremal elements are described as follows:
\begin{itemize}
\item $\small\Phi^{2\rightarrow 2}_0(A)=A$,
\item $\small\displaystyle \Phi^{2\rightarrow 2}_2(A)= \frac{1}{2}\left [\begin{array}{ccc} a_{00}+a_{11}&a_{12}&-a_{02}\\ a_{21}&a_{00}+a_{22}&a_{01}\\
-a_{20}&a_{10} & a_{11}+a_{22} \end{array} \right ]$,
\item \footnotesize$\displaystyle  \Phi^{2\rightarrow 2}_4(A)= \frac{1}{10}\left [\begin{array}{ccc} a_{00}+3a_{11}+6a_{22}&-2a_{01}-3a_{12}&a_{02}\\ -2a_{10}-3a_{21}&3a_{00}+4a_{11}+3a_{22}&-3a_{01}-2a_{12}\\
a_{20}&-3a_{10}-2a_{21} & 6a_{00}+3 a_{11}+a_{22} \end{array} \right ]$.\normalsize
\end{itemize}

For each case, the corresponding Kraus operators are given by 
	\begin{itemize}
		\item $K_1= \text{Id}_{3} $,
		\item  $ \displaystyle K_2=\frac{1}{\sqrt{2}}\left
		[\begin{array}{ccc} 0&0&0\\1&0&0\\0&1&0 \end{array} \right ]$ , $ \displaystyle K_3=\frac{1}{\sqrt{2}}\left
		[\begin{array}{ccc} -1&0&0\\0&0&0\\0&0&1 \end{array} \right ]$ , $K_4=-K_2^*$,\\
		\item $ \displaystyle K_5=\frac{1}{\sqrt{10}}\left
		[\begin{array}{ccc} 0&0&0\\0&0&0\\\sqrt{6}&0&0 \end{array} \right ]$ , $ \displaystyle K_6=\frac{1}{\sqrt{10}}\left
		[\begin{array}{ccc} 0&0&0\\-\sqrt{3}&0&0\\0&\sqrt{3}&0 \end{array} \right ]$, 

$ K_7=\displaystyle \frac{1}{\sqrt{10}}\left
		[\begin{array}{ccc} 1&0&0\\0&-2&0\\0&0&1 \end{array} \right ]$, $K_8=-K_6^*$, $K_9=K_5^*$.
	\end{itemize}

\begin{remark}
The convex set $\text{CovQC}(2,2)$ contains unitary conjugates of the well-known Werner-Holevo quantum channels. Let $\mathcal{W}_{sym}$ and $\mathcal{W}_{asym}$ denote the extremal Werner-Holevo quantum channels on $B(\Comp^3)$, i.e. 
\[\left \{ \begin{array}{ll}\mathcal{W}_{sym}(X)=\frac{1}{4}\left ( \text{Tr}(X)\text{Id}_3+X^t \right ),\\
\mathcal{W}_{asym}(X)=\frac{1}{2}\left ( \text{Tr}(X)\text{Id}_3- X^t \right ),\;\; X\in B(\Comp^3). \end{array} \right .
\] Then, for the unitary $R_2=\left [ \begin{array}{ccc}0&0&1\\ 0&-1&0\\ 1&0&0\end{array} \right ]$, we can see that the channels
\[\left \{ \begin{array}{ll} \Psi_1(X)=R_3^*\mathcal{W}_{sym}(X)R_3 =\frac{1}{6}\Phi^{2\rightarrow 2}_0(X)+\frac{5}{6}\Phi^{2\rightarrow 2}_4(X) \\ \Psi_2(X)=R_3^*\mathcal{W}_{asym}(X)R_3= \Phi^{2\rightarrow 2}_2(X), \;\; X\in B(\Comp^3) \end{array} \right .\]
are contained in $\text{CovQC}(2,2)$. Moreover, the linear space of $SU(2)_{(2,2)}$-covariant maps is spanned by the identity channel and $\Psi_1$, $\Psi_2$ (unitary conjugates of Werner-Holevo channels) thanks to linear independence of the set $\left \{ \Phi^{2\rightarrow 2}_0,\Phi^{2\rightarrow 2}_2, \Phi^{2\rightarrow 2}_4\right\}$.

\end{remark}

\section{EBT and PPT}

In this section we investigate EBT property of irreducibly $SU(2)$-covariant channels of low rank. It turns out that EBT property coincides with the closely related concept PPT property and we can get a full characterization in terms of the associated parameters. Recall that a quantum channel $\Phi:B(H_A)\rightarrow B(H_B)$ is called {\it EBT (entanglement-breaking)} if the corresponding  Choi matrix
    $$C_{\Phi}=\sum_{i,j=1}^{d_A}|i\ra \la j| \otimes \Phi(|i\ra \la j|)\in B(H_A \otimes H_B)$$
is {\it separable}, i.e. there exists a probability distribution $(p_i)_i$ and product quantum states $\rho^A_i\otimes \rho^B_i$ such that
\[\frac{1}{d_A}C_{\Phi}=\sum_i p_i \rho^A_i\otimes \rho^B_i.\]
Recall also that $\Phi$ is said to be {\it PPT (positive partial transpose)} if $(\text{id}\otimes T_B)(C_{\Phi})$ is positive where $T_B$ is the transpose map on $B(H_B)$. We will prove that EBT $=$ PPT in $\text{CovQC}(1,l)$, $\text{CovQC}(l,1)$ and $\text{CovQC}(2,2)$, respectively. Note that similar phenomena have been studied for covariant quantum channels with respect to $U(n)$ or $O(n)$ symmetries \cite{VW01}.

We first consider the case of $\text{CovQC}(1,l)$.

\begin{theorem}\label{thm-PPT-EBT1} Let $\Phi_p\in \text{CovQC}(1,l)$ from \eqref{eq-Cov(1,l)-ch} for $0\leq p\leq 1$ and $l\in \n$. Then, we have
    $$\text{$\Phi_p$ is PPT} \Leftrightarrow \frac{1}{l+1}\leq p\leq 1 \Leftrightarrow \text{$\Phi_p$ is EBT}.$$
\end{theorem}
\begin{proof}
Let us determine the range of $p$ for $\Phi_p$ being PPT. Observe that partial transposes of Choi matrices of $\Phi^{1\rightarrow l}_{l-1}$ and $\Phi^{1\rightarrow l}_{l+1}$ are simultaneously diagonalizable, and moreover, the partial transpose of $\Phi_p$ has only two eigenvalues
\begin{itemize}
\item $\displaystyle (1-p)\cdot \frac{2}{l+1}+p\cdot \frac{2}{(l+1)(l+2)}$ with multiplicity $l+2$ and
\item $\displaystyle (1-p)\cdot \left ( -\frac{2}{l(l+1)}\right )+p\cdot \frac{2}{l+1}$ with multiplicity $l$.
\end{itemize}
For the record, the associated eigenvectors are
\begin{itemize}
\item $|0,0\ra$, $|1,l\ra$ and $\displaystyle \frac{1}{\sqrt{l+1}}\left (\sqrt{l-s}|0, s+1\ra + \sqrt{s+1}|1,s\ra \right )$ with $0\leq s\leq l-1$ and
\item $\displaystyle \frac{1}{\sqrt{l+1}}\left (\sqrt{s+1}|0, s+1\ra - \sqrt{l-s}|1,s\ra \right )$ with $0\leq s\leq l-1$,
\end{itemize}
respectively. Thus, $\Phi_p$ is PPT if and only if \\
\begin{align}
 \displaystyle (1-p)\cdot \left ( -\frac{2}{l(l+1)}\right )+p\cdot \frac{2}{l+1}\geq 0,
 \nonumber
\end{align}
which leads us to the range of $p$ we wanted.

Now we are left to show that PPT implies EBT. Since we know that PPT channels in $\text{CovQC}(1,l)$ is another line segment, it is enough to check whether the end points $\Phi^{1\rightarrow l}_{l+1}$ and $\displaystyle \frac{l}{l+1}\Phi^{1\rightarrow l}_{l-1}+\frac{1}{l+1}\Phi^{1\rightarrow l}_{l+1}$ are EBT. From \cite[Theorem 5.7]{BCLY20} we get EBT of $\Phi^{1\rightarrow l}_{l+1}$, and the other case comes from the following averaging argument:
\begin{align*}
&\int_{SU(2)}\overline{\pi_m (x)}|i_1\ra\la i_1|\pi_m(x)^t\otimes \pi_l(x)|i_2\ra\la i_2|\pi_l(x)^* dx\\
& =\sum_{r=0}^{\min \left \{l,m\right\}}\frac{\sum_{i=0}^{|l-m|+2r}|C^{m,l,|l-m|+2r}_{m-i_1,i_2,i}|^2}{|l-m|+2r+1}(R_m\otimes \text{Id})\cdot p^{m,l}_{|l-m+2r|}\cdot (R_m\otimes \text{Id})^*,
\end{align*}
  where $p^{m,l}_{|l-m|+2r}$ is the orthogonal projection from $H_m\otimes H_l$ onto $H_{|l-m+2r|}$ and 
    \begin{equation}\label{eq-R_m}
        R_m=\sum_{j=0}^m (-1)^j |m-j\ra\la j|\in \mathcal{U}(m+1).
    \end{equation}
Then \cite[Proposition 4.5]{Al14} states that the above should a normalized Choi matrix $\frac{1}{m+1}C_{\Psi}$ of a quantum channel $\Psi$ given by 
\begin{equation}\label{eq-averaging}
\Psi=\sum_{r=0}^{\min \left \{l,m\right\}}\left (\sum_{i=0}^{|l-m|+2r}|C^{m,l,|l-m|+2r}_{m-i_1,i_2,i}|^2\right )\Phi^{m\rightarrow l}_{|l-m|+2r}. 
\end{equation}

When we put $m=1$, $i_1 = 1$, $i_2 = 0$ in \eqref{eq-averaging} we get
    $$\Psi = \displaystyle \frac{l}{l+1}\Phi^{1\rightarrow l}_{l-1}+\frac{1}{l+1}\Phi^{1\rightarrow l}_{l+1},$$
which is EBT by the above integral formula.

\end{proof}

Recall that positivity of a linear map $\Phi$ transfers to the adjoint map $\Phi^*$, then the case of $\text{CovQC}(l,1)$ follows immediately from \ref{subsubsec-(l,1)}.

\begin{theorem} Let $\Psi_p\in \text{CovQC}(l,1)$ from \eqref{eq-Cov(l,1)-ch} for $0\leq p\leq 1$ and $l\in \n$. Then we have
    $$\text{$\Psi_p$ is PPT} \Leftrightarrow \frac{1}{l+1}\leq p\leq 1 \Leftrightarrow \text{$\Psi_p$ is EBT}.$$
\end{theorem}

Now we turn our attention to the case $\text{CovQC}(2,2)$.

\begin{theorem}
Let $\mc N_{p,q}\in \text{CovQC}(2,2)$ from \eqref{eq-Cov(2,2)-ch} for $p,q\geq 0$ with $p+q\leq 1$. Then we have
    $$\text{$\mc N_{p,q}$ is PPT} \Leftrightarrow \begin{cases}
    0\leq p\leq \frac{1}{2}\\ \frac{2}{3}\leq p+q\leq 1\end{cases} \Leftrightarrow \text{$\mc N_{p,q}$ is EBT}.$$
\end{theorem}
\begin{proof}
We basically follow the same approach as in Theorem \ref{thm-PPT-EBT1}. We observe that the partial transpose of the Choi matrix of $\mc N_{p,q}$ has the following eigenvalues:
\begin{itemize}
\item $1-2p$ with multiplicity $1$,
\item $1-\frac{1}{2}p-\frac{9}{10}q$ with multiplicity $5$,
\item $\frac{1}{2}(3p+3q-2)$ with multiplicity $3$.
\end{itemize}
For the record, the associated eigenvectors are \small
\begin{itemize}
\item $\displaystyle \frac{1}{\sqrt{3}} \left( |02\ra-|11\ra+|20\ra \right)$,
\item $|00\ra$, $|22\ra$, $\displaystyle \frac{1}{\sqrt{2}} \left( |01\ra+|10\ra \right)$, $\displaystyle \frac{1}{\sqrt{2}} \left( |12\ra+|21\ra \right)$, $\displaystyle \frac{1}{\sqrt{6}} \left( |02\ra+2|11\ra+|20\ra \right)$,
\item $\displaystyle \frac{1}{\sqrt{2}} \left( |01\ra-|10\ra \right)$, $\displaystyle \frac{1}{\sqrt{2}} \left( |12\ra-|21\ra \right)$, $\displaystyle \frac{1}{\sqrt{2}} \left( |02\ra-|20\ra \right)$,
\end{itemize}
\normalsize respectively. We get the wanted range of $(p,q)$ for $\mc N_{p,q}$ being PPT  by requiring all eigenvalues to be non-negative.

Now we need to show that PPT implies EBT. The convex set of all PPT elements in $\text{CovQC}(2,2)$ has extreme points $\mc N_{p,q}$ for $(p,q)$ being one of $(0,1)$, $\displaystyle \left (\frac{1}{2},\frac{1}{2}\right )$, $\displaystyle \left (0,\frac{2}{3}\right )$, $\displaystyle \left (\frac{1}{2},\frac{1}{6}\right )$. From \eqref{eq-averaging} the Choi matrix of the above extremal channel $\mc N_{p,q}$ is given by
\begin{equation}
3\cdot \int_{SU(2)} \overline{\pi_2(x)}|i_1\ra\la i_1| \pi_2(x)^t\otimes \pi_2(x)|i_2\ra\la i_2| \pi_2(x)^*dx
\end{equation}
with $|i_1i_2\ra=|22\ra,|21\ra,|11\ra,|20\ra$ for each case.

\end{proof}

\begin{center}
\begin{figure}[hbt!]
    \centering
	\includegraphics[scale=0.17]{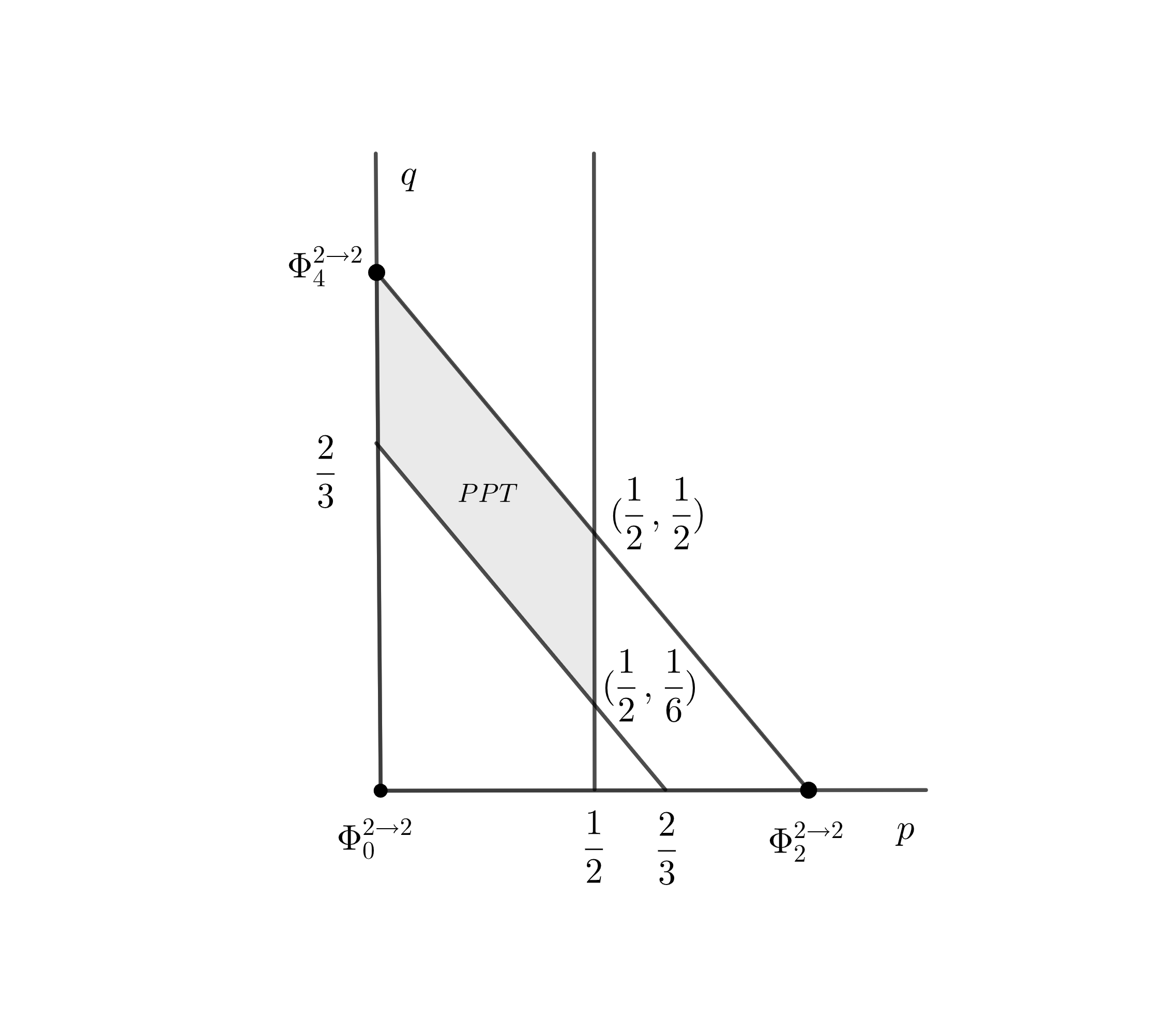}
	\caption{The region of PPT channels in $\text{CovQC}(2,2)$}
\end{figure}
\end{center}

\section{Degradability}\label{sec-degradability}

Recall that a quantum channel $\Phi: B(H_A) \to B(H_B)$ with a complementary channel $\Phi^c$ is called {\it degradable} if there is another quantum channel $\Psi$ such that $\Psi \circ \Phi = \Phi^c$. Note that degradability is independent of the choice of complementary channels, which we may have multiple realizations. In this article we will use the realization coming from a Kraus representation of the original channel. More precisely, a Kraus representation
    $$\Phi(X) = \sum^m_{i=1}K_i X K^*_i,\;\; K_i \in B(H_A, H_B),\;\; 1\le i \le m$$
yields a complementary channel $\Phi^c: B(H_A) \to B(\Comp^m)$ given by     
    $$\Phi^c(X) = \left[\text{Tr}(K_iXK^*_j)\right]^m_{i,j=1}.$$

Since degradability plays crucial roles in QIT as a sufficient condition for additivity of coherent information and private information, the structure of degradable quantum channels has been studied in various contexts, particularly for low dimensional cases. See \cite{NG98,Ce00,CRS08} for more details.
In the case of irreducibly $SU(2)$-covariant quantum channels, \cite{BCLY20} investigated degradability of some extremal elements in $\text{CovQC}(k,l)$. In particular, the $SU(2)$-Clebsch-Gordan channel $\Phi^{1\rightarrow l}_{l-1}$ is degradable whilst $\Phi^{1\rightarrow l}_{l+1}$ is not (\cite[Theorem 5.8]{BCLY20}). 

In this section, we examine degradability of irreducibly $SU(2)$-covariant channels of low rank, namely $\text{CovQC}(1,l)$, $\text{CovQC}(l,1)$ and $\text{CovQC}(2,2)$. It turns out that degradability holds only in some of the extreme points. This result can be considered a generalization of the fact that qubit depolarizing quantum channels are not degradable except for the noiseless case.
We begin with the case $\text{CovQC}(1,l)$.

\begin{theorem}\label{thm-deg1}
The channel $\Phi_p \in \text{CovQC}(1,l)$ from \eqref{eq-Cov(1,l)-ch} is degradable only for $p=0$. 
\end{theorem}

\begin{proof}
Note that degradability at $p=0$ follows from \cite[Theorem 5.8]{BCLY20}, so it is enough to show that $\Phi_p$ is not degradable for all $p>0$. Now we assume that there exists a quantum channel $\Psi$ satisfying $\Phi^c_p=\Psi\circ\Phi_p$, where $\Phi^c_p$ is the one from the choice of Kraus operators in \ref{subsubsec-(1,l)}. From \eqref{eq-l-1} we have
\begin{align} 
\label{eq:5-1}&\Psi \left (\sum_{i=0}^{l} \frac{2}{l+1}|i\ra\la{i}|\right ) = \Phi^c(|0\ra\la{0}|)+\Phi^c(|1\ra\la{1}|) \\
\label{eq:5-2} &\Psi\left (\sum_{i=0}^{l}(a_i-a_{l-i})|i\ra\la{i}|\right ) = \Phi^c(|0\ra\la{0}|)-\Phi^c(|1\ra\la{1}|),
\end{align}
where $\displaystyle a_i=(1-p)\cdot \frac{2(l-i)}{l(l+1)} + p\cdot  \frac{2(i+1)}{(l+1)(l+2)}$, $1\le i \le l$.

Suppose $\displaystyle a_0-a_l\ne 0 \left (\Leftrightarrow p\neq\frac{l+2}{2(l+1)}\right )$, then taking an appropriate linear combination we get\small
\begin{align*}
&\displaystyle \Psi\left( 2\cdot|0\ra\la{0}|+\left( 1+\sum_{i=1}^{l}\frac{a_i-a_{l-i}}{a_0-a_{l}}\right) |i\ra\la{i}|\right) = \\
&\frac{(l+1)(l+2)}{2\big((l+2)-2(l+1)p\big)}\bigg(\Big(2-\frac{2(l+1)p}{l+2}\Big)\Phi^c(|0\ra\la{0}|)-\frac{2(l+1)p}{l+2}\Phi^c(|1\ra\la{1}|)\bigg)
\end{align*}
\normalsize
We can see that the input of $\Psi$ in the above is positive %({\bf semi-definite?}) 
since $\displaystyle 1+\frac{a_i-a_{l-i}}{a_0-a_l}=1+\frac{l-2i}{l}\geq0$ for all $0\leq i \leq l$, but the output is not. Indeed, it is immediate from \eqref{eq-Kraus2} to check that
$(2l+2,2l+2)$-entries of $\Phi^c(|0\ra\la 0|)$ and $\Phi^c(|1\ra\la 1|)$ are $0$ and $\displaystyle \frac{2p}{l+2}$, respectively, which explains the output matrix is not positive for any $p>0$.
%$(l+1,l+1)$-entries of $\Phi^c(|0\ra\la 0|)$ and $\Phi^c(|1\ra\la 1|)$ are $\displaystyle \frac{2p}{l+2}$ and $0$, whilst $(2l+2,2l+2)$-entries are $0$ and $\displaystyle \frac{2p}{l+2}$, respectively. This explains the output matrix is not positive for $p>0$.

Now the remaining case is when $\displaystyle a_0-a_l= 0 \left (\Leftrightarrow p=\frac{l+2}{2(l+1)}\right )$. We can easily see that $ \displaystyle a_i=\frac{1}{l+1}$ for all $0\leq i\leq l$, so that \eqref{eq:5-2} implies $\Phi^c(|0\ra\la{0}|)=\Phi^c(|1\ra\la{1}|)$, which we already have seen not to be true for $p>0$.

\end{proof}

We have a parallel result in $\text{CovQC}(l,1)$ as follows, though the proof is not direct from Theorem \ref{thm-deg1}.

\begin{theorem}
    The channel $\Psi_p \in \text{CovQC}(l,1)$ from \eqref{formula-Cov(l,1)-ch} is degradable only for $p=0$. 
\end{theorem}

\begin{proof}
	Degradability for the case $p=0$ follows from \cite[Theorem 5.8]{BCLY20}, so let us show non-degradability for all $p>0$. Recall from Section \ref{subsubsec-(l,1)} that the Kraus operators $L_1,L_2,\cdots,L_l$ of $\Phi^{l\rightarrow 1}_{l-1}$ and $L_{l+1},\cdots,L_{2l+2}$ of $\Phi^{l\rightarrow 1}_{l+1}$ satisfy 
\begin{itemize}
\item $L_i|0\ra=\delta_{i,l}(-1)^{l-1}|0\ra$ and $L_i|l\ra=\delta_{i,1} |1\ra\la 1|$ for any $1\leq i\leq l$,
\item $L_i|0\ra=0$ for all $l+1\leq i\leq 2l$, $L_{2l+1}|0\ra=\frac{(-1)^l}{\sqrt{l+2}}|0\ra$ and $L_{2l+2}|0\ra=\frac{(-1)^l\sqrt{l+1}}{\sqrt{l+2}}|1\ra$,
\item $L_{l+1}|l\ra=\frac{\sqrt{l+1}}{\sqrt{l+2}}|0\ra$,  $L_{l+2}|l\ra=\frac{1}{\sqrt{l+2}}|1\ra$ and $L_i|l\ra=0$ for all $l+3\leq i\leq 2l+2$.
\end{itemize}
The above tell us that, for $\Psi_p=(1-p)\Phi^{1\rightarrow l}_{l-1}+p\Phi^{1\rightarrow l}_{l+1}$, we have
\begin{align*}
    \Psi_p(|0\ra\la 0|)&=(1-p)|0\ra\la 0|+p\left ( \frac{1}{l+2}|0\ra\la 0|+\frac{l+1}{l+2}|1\ra\la 1| \right )\\
    &=\left (1- \frac{p (l+1)}{l+2}\right )|0\ra\la 0|+\frac{p(l+1)}{l+2}|1\ra\la 1|,\\
    \Psi_p(|l\ra\la l|)&=(1-p)|1\ra\la 1|+p \left ( \frac{l+1}{l+2}|0\ra\la 0|+\frac{1}{l+2}|1\ra\la 1| \right )\\
    &=\frac{p(l+1)}{l+2}|0\ra\la 0|+\left (1- \frac{p (l+1)}{l+2}\right ) |1\ra\la 1|
\end{align*}

On the other side, we have the following facts on the complementary channel $\Psi_p^c:B(H_l)\rightarrow M_{2l+2}(\Comp)$ from the above information of the operators $L_i$:

\begin{itemize}
    \item $[\Psi_p^c(|0\ra\la 0|)]_{i,i}=0$ for all $l+1\leq i\leq 2l$, $[\Psi_p^c(|0\ra\la 0|)]_{2l+1,2l+1}=\displaystyle   \frac{p}{l+2}$ and $[\Psi_p^c(|0\ra\la 0|)]_{2l+2,2l+2}=\displaystyle \frac{p(l+1)}{l+2}$,

\item $[\Psi_p^c(|l\ra\la l|)]_{l+1,l+1}=\displaystyle \frac{p(l+1)}{l+2}$, $[\Psi_p^c(|l\ra\la l|)]_{l+2,l+2}=\displaystyle \frac{p}{l+2}$ and we have $[\Psi_p^c(|l\ra\la l|)]_{i,i}=0$ for all $l+3\leq i\leq 2l+2$.
\end{itemize}

Now, if we assume $\Psi_p$ is degradable, then there should be a quantum channel $\mathcal{L}:B(H_1)\rightarrow M_{2l+2}(\Comp)$ satisfying $\mathcal{L}\circ \Psi_p=(\Psi_p)^c$. In particular, for a specific input matrix $X=\displaystyle \left ( 1-\frac{p(l+1)}{l+2}\right )|0\ra\la 0|-\frac{p(l+1)}{l+2}|l\ra\la l|$, we should have
\begin{equation}\label{eq-deg2}
    (\Psi_p)^c(X)=\mathcal{L}(\Psi_p(X))=\mathcal{L}\left ( 
    \left ( 1-\frac{2p(l+1)}{l+2}\right ) |0\ra\la 0| \right ).
\end{equation}

If $\displaystyle 1-\frac{2p(l+1)}{l+2}=0$, then the equation \eqref{eq-deg2} with $X=\frac{1}{2}|0\ra\la 0| +\frac{1}{2}|l\ra\la l|$ implies $(\Psi_p)^c(|0\ra\la 0|)=(\Psi_p)^c(|l\ra\la l|)$, which is not true due to the above observations for the complementary channel $\Psi_p^c$. For $\displaystyle 1-\frac{2p(l+1)}{l+2}>0$ positivity of $\mathcal{L}$ imply that all diagonal entries of $(\Psi_p)^c(X)$ are non-negative, but the $(l+1,l+1)$-th entry is $-\displaystyle \left ( \frac{p(l+1)}{l+2} \right )^2 <0$, and for $\displaystyle 1-\frac{2p(l+1)}{l+2}<0$ positivity of $\mathcal{L}$ imply that all diagonal entries of $(\Psi_p)^c(X)$ are non-positive, but the $(2l+2,2l+2)$-th entry is $-\displaystyle \left (1- \frac{p(l+1)}{l+2} \right )\cdot \frac{p(l+1)}{l+2}>0$. Thus, we get contradiction for all cases.

\end{proof}

In case of $\text{CovQC}(2,2)$ we have only two degradable quantum channels, whose degradability was already noted in \cite{BCLY20}:

\begin{theorem}\label{thm-deg2}
	The channel $\mc N_{p,q}\in \text{CovQC}(2,2)$ from \eqref{eq-Cov(2,2)-ch} for $p,q\geq 0$ with $p+q\leq 1$ is degradable only when $(p,q)=(0,0)$ or $(p,q)=(1,0)$.
\end{theorem}
\begin{proof}
	Degradability of the cases  $(p,q)=(0,0)$ or $(p,q)=(1,0)$ follows from the fact that  $\Phi^{2\rightarrow 2}_0=\text{id}_3$ and $\Phi^{2\rightarrow 2}_2=(\Phi^{2\rightarrow 2}_2)^c$ or \cite[Theorem 5.8]{BCLY20}. Now let us prove that $\mathcal{N}_{p,q} = (1-p-q)\Phi^{2\rightarrow 2}_0+p\Phi^{2\rightarrow 2}_2+q\Phi^{2\rightarrow 2}_4 $ is non-degradable for all the other cases. Assume there exists a linear map $\Psi$ satisfying $\mathcal{N}_{p,q}^c=\Psi\circ\Phi$. From the observations in Section \ref{subsubsec-(2,2)}, we have
	\begin{align*}
	&\displaystyle \mathcal{N}_{p,q}\left( |0\ra\la{0}|\right) = \Big( 1-\frac{p}{2}-\frac{9q}{10}\Big) |0\ra\la{0}|+
	\Big(\frac{p}{2}+\frac{3q}{10}\Big)|1\ra\la{1}|+\frac{6q}{10}|2\ra\la{2}| \\ 
	&\mathcal{N}_{p,q}\left( |1\ra\la{1}|\right) = \Big(\frac{p}{2}+\frac{3q}{10}\Big)|0\ra\la{0}|+
	\Big(1-p-\frac{6q}{10}\Big)|1\ra\la{1}|+\Big(\frac{p}{2}+\frac{3q}{10}\Big)|2\ra\la{2}|\\
	&\mathcal{N}_{p,q}\left( |2\ra\la{2}|\right) = \frac{6q}{10}|0\ra\la{0}|+
	\Big(\frac{p}{2}+\frac{3q}{10}\Big)|1\ra\la{1}|+\Big(1-\frac{p}{2}-\frac{9q}{10}\Big)|2\ra\la{2}|,
	\end{align*}
	and explicit Kraus operators $L_1,\cdots,L_9$ of $\mathcal{N}_{p,q}$ given by 
	\[\sqrt{1-p-q}K_1,\sqrt{p}K_2,\cdots,\sqrt{p}K_4,\sqrt{q}K_5,\cdots,\sqrt{q}K_9\]
	where $K_j$'s are described in Subsection \ref{subsubsec-(2,2)}. In particular, we have 
	\begin{align*}
	\mathcal{N}_{p,q}^c(|0\ra\la 0|) =[(L_j^*L_i)_{0,0}]_{i,j=1}^9\\
	\mathcal{N}_{p,q}^c(|1\ra\la 1|) =[(L_j^*L_i)_{1,1}]_{i,j=1}^9\\
	\mathcal{N}_{p,q}^c(|2\ra\la 2|) =[(L_j^*L_i)_{2,2}]_{i,j=1}^9
	\end{align*}
	
Taking appropriate linear combinations, we have the following, \footnotesize
	\begin{align} 
	&M:=-\Big(\frac{p}{2}+\frac{3q}{10}\Big)\mathcal{N}_{p,q}^c(|0\ra\la{0}|)+
	\Big(1-\frac{p}{2}-\frac{3q}{10}\Big)\mathcal{N}_{p,q}^c(|1\ra\la{1}|)
	-\Big(\frac{p}{2}+\frac{3q}{10}\Big)\mathcal{N}_{p,q}^c(|2\ra\la{2}|)\nonumber \\
	&=(\Psi\circ \mathcal{N}_{p,q})\left ( - \left (\frac{p}{2}+\frac{3q}{10}\right )|0\ra\la 0|  + \left (1-\frac{p}{2}-\frac{3q}{10}\right ) |1\ra\la 1| - \left (\frac{p}{2}+\frac{3q}{10}\right )|2\ra\la 2| \right )\nonumber\\
	&=\Psi\left (\left (1-\frac{3p}{2}-\frac{9q}{10}\right )|1\ra\la{1}|\right )\nonumber
	\end{align}
	\normalsize	and also \footnotesize
	\begin{align} 
	&N:=-\Big(1-p-\frac{6q}{10}\Big)\Phi^c(|0\ra\la{0}|)+
	\Big(p+\frac{6q}{10}\Big)\Phi^c(|1\ra\la{1}|)
	-\Big(1-p-\frac{6q}{10}\Big)\Phi^c(|2\ra\la{2}|) \nonumber\\
	&=(\Psi\circ \Phi)\left ( -\left (1-p-\frac{6q}{10}\right )|0\ra\la 0|  + \left (p+\frac{6q}{10}\right ) |1\ra\la 1| - \left (1-p-\frac{6q}{10}\right )|2\ra\la 2| \right )\nonumber\\
	&=\Psi\left ( -\left( 1-\frac{3p}{2}-\frac{9q}{10}\right )\left(|0\ra\la{0}|+|2\ra\la{2}|\right) \right)\nonumber
	\end{align}
	\normalsize	Here we may assume that  $ \displaystyle 1-\frac{3p}{2}-\frac{9q}{10} \neq 0$, otherwise we should have $M=0$, i.e.
	\[-\frac{1}{3}\mathcal{N}_{p,q}^c(|0\ra\la{0}|)+
	\frac{2}{3}\mathcal{N}_{p,q}^c(|1\ra\la{1}|)
	-\frac{1}{3}\mathcal{N}_{p,q}^c(|2\ra\la{2}|)=0,\]
which is equivalent to $\displaystyle \frac{1}{3}\cdot \mathcal{N}_{p,q}^c(\text{Id}_3)=\mathcal{N}_{p,q}^c(|1\ra\la{1}|)$. Then entrywise comparison should give us $p=q=0$, a contradiction.

	Suppose $ \displaystyle 1-\frac{3p}{2}-\frac{9q}{10} > 0$. Note that all the eigenvalues of $M$ should be nonnegative if $\Psi$ is positive.  However, using the explicit outcomes of the complementary channel, it is immediate that $\displaystyle \la4|M|4\ra=-\frac{6p}{10}(\frac{p}{2}+\frac{3q}{10})<0 $.
	Similarly, if $ \displaystyle 1-\frac{3p}{2}-\frac{9q}{10} < 0$, then $N$ should be positive if $\Psi$ is positive. However, $\displaystyle \la8|N|8\ra = -\frac{6q}{10}\Big(1-p-\frac{6q}{10}\Big)<0$ shows that $\Psi$ cannot be a positive map in either case.
	
	%Recall that output states of $\Phi^c(\rho)$ with diagonal density matrices $\rho$ are described as block diagonal matrices as explained in Example \ref{ex3}. In particular, \eqref{eq:5-4} and \eqref{eq:5-5} has a block decomposition consisting of one $3\times 3$ matrix, two $2\times 2$ matrices and two $1\times 1$ matrices. In particular,  
	%$\displaystyle \lambda_1= -p\Big(\frac{p}{2}+\frac{3q}{10}\Big)<0 $ and $\displaystyle \lambda_2=-\frac{6q}{10}\Big(1-p-\frac{6q}{10}\Big)<0$ are $1\times 1$ block of \eqref{eq:5-4} and \eqref{eq:5-5} respectively. If $\displaystyle \Big( 1-\frac{3p}{2}-\frac{9q}{10}\Big ) >0$, then $\lambda_1<0$ fails to hold the positivity of $\Psi$. Otherwise if $\displaystyle \Big( 1-\frac{3p}{2}-\frac{9q}{10}\Big ) <0$, $\lambda_2<0$ fails the positivity, and it completes the proof. 
\end{proof}

\section{Holevo information and MOE of $SU(2)$-covariant quantum channels}

For any quantum channel $\Phi: B(H_A) \to B(H_B)$ the Holevo information $\chi(\Phi)$ is defined by
$$\chi(\Phi)=\sup_{(p_i)_i,(\rho_i)_i}\left \{H\left ( \sum_{i=1}^n p_i \Phi(\rho_i)\right )  - \sum_{i=1}^n p_i H(\Phi(\rho_i)) \right \},$$
where the supremum runs over all probability distributions $(p_i)_{i=1}^n$ and families of quantum states $(\rho_i)_{i=1}^n$ in $B(H_A)$. The quantity $\chi(\Phi)$ is closely related to another quantity called the {\it minimum output entropy (shortly, MOE)} $H_{min}(\Phi)$ given by $$H_{min}(\Phi)=\min_{\rho} H(\Phi(\rho)),$$
where $\rho$ runs over all quantum states in $B(H_A)$. Note that the minimum is attained at a pure state thanks to concavity of the entropy function, which we use the natural logarithm for the definition. In general, we have
    $$\chi(\Phi) \le \log(d_B) - H_{min}(\Phi),$$
but there are cases when the equality holds in the above. For example, it is known that the equality holds, i.e. $\chi(\Phi) = \log(d_B) - H_{min}(\Phi)$ when $\Phi$ is an irreducibly $G$-covariant quantum channel \cite{Ho06}, which includes the class of channels we are investigating. This equality allows us to focus on a quantity, namely MOE, which is relatively easier to compute. In this section we will show that MOE can be easily calculated by identifying minimizers for all $SU(2)_{(1,l)}$-covariant channels. Note that the extremal $SU(2)_{(1,l)}$-covariant channels $\Phi^{1\rightarrow l}_{l-1}$ and $\Phi^{1\rightarrow l}_{l+1}$ are special cases of the channels $\Phi^{k\rightarrow l}_{|k-l|}$ and $\Phi^{k\rightarrow l}_{k+l}$ (in which the subindices are the highest and the lowest ones), and \cite{LiSo14} proved that the coherent state $|0\ra\la 0|$ is a minimizer for $H_{min}(\Phi^{k\rightarrow l}_{|k-l|})$ and $H_{min}(\Phi^{k\rightarrow l}_{k+l})$. In general, any density matrix of the form $\pi_k(x)|0\ra\la 0 | \pi_k(x)^*$ is called a (Bloch) coherent state in $B(H_k)$ \cite{ACGT72,Ho78,LiSo14}.

\begin{proposition} For the channel $\Phi_p \in \text{CovQC}(1,l)$ from \eqref{eq-Cov(1,l)-ch} the minimum output entropy of $\Phi_p$ is attained at the coherent state $|0\ra\la 0|$. Moreover, the Holevo information of $\Phi_p$ is given by
	\[\chi(\Phi_p)=\log(l+1)-H(\Phi_p(|0\ra\la 0|)).\]
\end{proposition}
\begin{proof}
Since any density matrix $\rho$ is written as $UDU^*$ with a diagonal density matrix $D\in \mathcal{D}(\Comp^2)$ and $U\in SU(2)$, we have
\begin{equation}
    \Phi_p(\rho)=\Phi_p(\pi_1(U)D\pi_1(U)^*)=\pi_l(U)\Phi_p(D)\pi_l(U)^*.
\end{equation}
This fact implies that diagonal density matrices are enough to compute MOE, i.e. 
$\displaystyle H_{min}(\Phi_p)=\min_{D\in \mathcal{D}(\Comp^2):\text{ diagonal}}H(\Phi_p(D))$. Moreover, the minimum should be attained at $|0\ra\la 0|$ or $|1\ra\la 1|$  since the entropy is concave on $\mathcal{D}(\Comp^2)$, and their output entropies are same, i.e. $H(\Phi(|0\ra\la 0|))=H(\Phi(|1\ra\la 1|))$ from \eqref{eq-l-1}. Thus, both $|0\ra\la 0|$ and $|1\ra\la 1|$ are minimizers for the minimum output entropy.
\end{proof}

\begin{remark}
The graphs of the minimum output entropy $H_{min}$ and the Holevo information $\chi$ for the cases $l=2,3,4$ are given as follows:

\begin{figure}[hbt!]
    \centering
\includegraphics[scale=0.26]{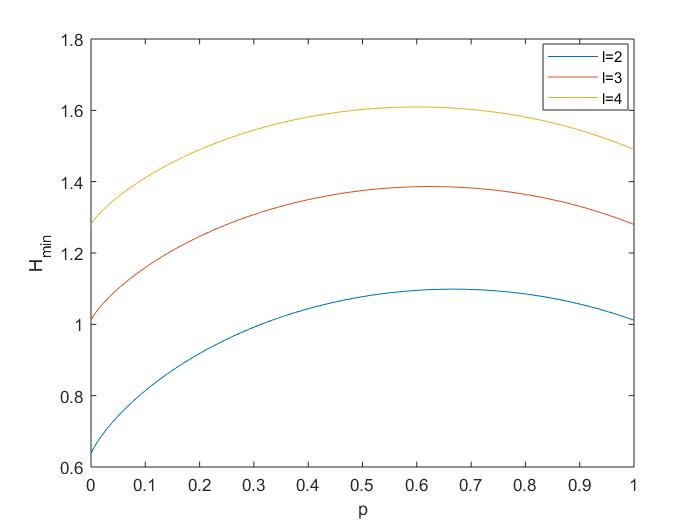}\includegraphics[scale=0.26]{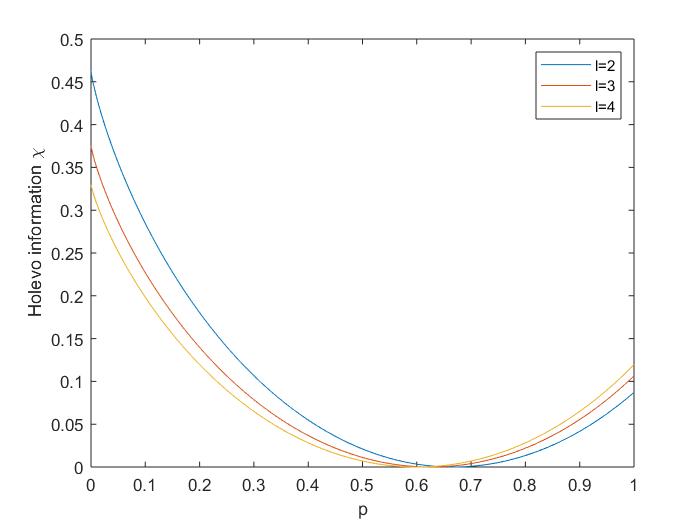}
\caption{$H_{min}$ and $\chi$ in $\text{CovQC}(1,l)$}
\end{figure}

In the right figure, the Holevo information equals to $0$ when $p=\displaystyle \frac{l+2}{2(l+1)}$, which corresponds to the completely depolarizing channel $\displaystyle \rho\mapsto \frac{\text{Tr}(\rho)}{l+1}\text{Id}_{l+1}$.
\end{remark}

The case $\text{CovQC}(2,2)$ is more involved since the coherent state $|0\ra\la 0|$ is not always the minimizer for $H_{min}(\mathcal{N}_{p,q})$, where $\mc N_{p,q}\in \text{CovQC}(2,2)$ is from \eqref{eq-Cov(2,2)-ch}. Indeed, for $(p,q)=$ (0.5,0.5), we have
$$H(\mathcal{N}_{p,q}(|1\ra\la 1|))\approx 1.055<1.089\approx H(\mathcal{N}_{p,q}(|0\ra\la 0|)),$$
and the state $|1\ra\la 1|$ is not a Bloch coherent state.

Fortunately, the list of possible minimizers for $H_{min}(\mathcal{N}_{p,q})$ does not extend beyond the two states, namely $|0\ra\la 0|$ and $|1\ra\la 1|$, thanks to a detailed analysis on the output eigenvalues.
The first step for the analysis is to single out a fixed eigenvalue $\displaystyle\frac{p}{2}+\frac{3q}{10}$ for a fixed pair $(p,q)$ regardless of the input pure state $|\xi\ra\la \xi|$. 

\begin{lemma}\label{lem51}
	Let $\mc N_{p,q}\in \text{CovQC}(2,2)$ from \eqref{eq-Cov(2,2)-ch} for $p,q\geq 0$ with $p+q\leq 1$. Then $\displaystyle\frac{p}{2}+\frac{3q}{10}$ is 
	an eigenvalue of $\mc N_{p,q} \left( |\xi\ra \la \xi | \right)  $ for any unit vector $ \xi \in \Comp^3$ .
\end{lemma}

\begin{proof}
Consider the following decomposition of $\mc N_{p,q}$.
    $$\mc N_{p,q}(X) = (1-m-n)X +m\cdot R_2X^t R_2^*+ \frac{n \text{Tr}(X)}{3}\cdot \text{Id}_3,\; X\in B(\Comp^3),$$
where $m=-\displaystyle\frac{p}{2}+\frac{3q}{10}$, $n=\displaystyle\frac{3p}{2}+\frac{9q}{10}$ and $R_2$ is from \eqref{eq-R_m}.
%$\displaystyle \Psi_1 = \frac{1}{3}\Phi_{0}^{2 \rightarrow 2}-\Phi_{2}^{2 \rightarrow 2}+\frac{5}{3}\Phi_{4}^{2 \rightarrow 2}$ and the completely depolarizing channel $\Psi_2$, i.e. $\displaystyle \Psi_2(X)=\frac{\text{Tr}(X)}{3} \text{Id}_3$, $X\in B(\Comp^3)$. Then, we have
%\begin{align*}
%\mc N_{p,q}(X)&= (1-m-n)\Phi_{0}^{2 \rightarrow 2}(X) +m\Psi (X) +n \Psi'(X)\\
%&=(1-m-n)X +m\cdot R_2X^t R_2^*+ \frac{n \text{Tr}(X)}{3}\cdot \text{Id}_3
%\end{align*}
%where $m=-\displaystyle\frac{p}{2}+\frac{3q}{10}$, $n=\displaystyle\frac{3p}{2}+\frac{9q}{10}$ and $R_2$ is from \eqref{eq-R_m}.
In order to show that $\mc N_{p,q}\left(|\xi\ra \la \xi |\right)$ has an eigenvalue $\displaystyle \frac{n}{3}$ it is sufficient to show that $0$ is an eigenvalue of the linear map
$(1-m-n)\cdot |\xi\ra\la \xi|+m \cdot R_2 |\overline{\xi}\ra\la \overline{\xi}|R_2^*,$
which is clear since its rank is less than or equal to $2$, so that the corresponding kernel space is nontrivial.
\end{proof}

\begin{theorem}
Let $\mc N_{p,q}\in \text{CovQC}(2,2)$ from \eqref{eq-Cov(2,2)-ch} for $p,q\geq 0$ with $p+q\leq 1$. Then we have
    $$H_{min}(\mathcal{N}_{p,q}) = \begin{cases} H(\mc N_{p,q}(|0\ra\la 0|)), & (5p-3q) (5p+6q-5)\leq 0,\\ H(\mc N_{p,q}(|1\ra\la 1|)), & (5p-3q) (5p+6q-5) \geq 0.\end{cases}$$
Moreover, the Holevo information $\chi(\mathcal{N}_{p,q})$ is given by $\log(3)-H_{min}(\mathcal{N}_{p,q})$.
\end{theorem}
\begin{proof}

Let $\{\lambda_j = \lambda_j(p,q,\xi): 0\le j \le 2 \}$ be the set of all eigenvalues (with repetition) of $\mc N_{p,q}\left(|\xi\ra \la \xi |\right)$. Let $\lambda_0=\displaystyle \frac{p}{2}+\frac{3q}{10}$ be the one we already specified.

For fixed $(p,q)$ we know that $\lambda_1(\xi)+\lambda_2(\xi)=1-\lambda_0$ is constant, which means that $H(\mc N_{p,q}(|\xi \ra\la \xi|))$ is minimized at $\xi = \xi_0$ if and only if $\lambda_1(\xi)$ and $\lambda_2(\xi)$ are farthest from each other at $\xi = \xi_0$ if and only if $\lambda_1(\xi) \lambda_2(\xi)$ is minimized at $\xi = \xi_0$. This observation allows us to focus on the function
    $$f(p,q,|\xi\ra) := \lambda_1(p,q,\xi) \lambda_2(p,q,\xi)$$
and its companions
    $$\begin{cases}D_0(p,q,|\xi\ra) = f(p,q,|\xi\ra)-f(p,q,|0\ra),\\
D_1(p,q,|\xi\ra)=f(p,q,|\xi\ra)-f(p,q,|1\ra).\end{cases}$$
Now it is enough to show that at least one of $D_0$ or $D_1$ is non-negative for any $(p,q,|\xi\ra)$. In order to achieve that goal we fix $(q,|\xi\ra)$ and consider the behavior of $p\mapsto f(p,q,|\xi\ra)$. Note that $f(p,q,|\xi\ra)$ is a quadratic polynomial of $p$ with a complicated formula, but its leading coefficient is relatively easy to describe. Indeed, $\text{det}\left (\mathcal{N}_{p,q}(|\xi\ra\la \xi|)\right )$ is a cubic polynomial since each entry of $\mathcal{N}_{p,q}(|\xi\ra\la \xi|)$ is a linear function of $p$ and $q$. Moreover, the determinant is $\displaystyle \sum_{\sigma\in S_3}\text{sgn}(\sigma)x_{1\sigma(1)}x_{2\sigma(2)}x_{3\sigma(3)}$ where the matrix $(x_{ij})_{1\leq i,j\leq 3}$ denotes $(1-p-q)\Phi^{2\rightarrow 2}_0(|\xi\ra\la \xi|)+p\Phi^{2\rightarrow 2}_2(|\xi\ra\la \xi|)+q\Phi^{2\rightarrow 2}_4(|\xi\ra\la \xi|)$, so the leading coefficient (as a polynomial of $p$) of $\text{det}(\mathcal{N}_{p,q}(|\xi\ra\la \xi|)$ is exactly the same as the leading coefficient of 
\[\text{det} \left (-p\Phi^{2\rightarrow 2}_0(|\xi\ra\la \xi|)+p\Phi^{2\rightarrow 2}_2(|\xi\ra\la \xi|) \right )=p^3\text{det} \left (-|\xi\ra\la \xi|+\Phi^{2\rightarrow 2}_2(|\xi\ra\la \xi|) \right ).\]
Dividing by the fixed eigenvalue $\lambda_0=\frac{p}{2}+\frac{3q}{10}$ we get the leading coefficient of $f(p,q,|\xi\ra)$ as
\[2\cdot \text{det} \left (-|\xi\ra\la \xi|+\Phi^{2\rightarrow 2}_2(|\xi\ra\la \xi|) \right )=-\frac{|b^2-2ac|^2}{2}\]
where $\xi=[a, b, c]^t\in \Comp^3$, which implies $\displaystyle \frac{\partial^2 f}{\partial p^2}=-|b^2-2ac|^2$.

%More precisely, we have \textcolor{red}{\bf (Fill in the details!)}
    %$$\frac{\partial^2 f}{\partial p^2}=-|b^2-2ac|^2,\;\; \xi=[a, b, c]^t\in \Comp^3.$$
Note that we can easily exclude the cases $|\xi\ra = |0\ra, |2\ra$ since $|2\ra$ is also a Bloch coherent state, so that we may assume $0<|b^2-2ac|<1$. Thus, we have $\displaystyle \frac{\partial ^2 D_0}{\partial p^2}=-|b^2-2ac|^2 < 0$
and $\displaystyle \frac{\partial ^2 D_1}{\partial p^2}=1-|b^2-2ac|^2 > 0,$
which means that $D_0$ is concave and $D_1$ is convex as functions of $p$. Now the strategy is to locate zeros of $D_0$ and $D_1$, so that we can secure the intervals where either of functions have non-negative values. Indeed, we have
$D_0(p)=0=D_1(p)$ for $\displaystyle p=\frac{3}{5}q$ and $\displaystyle p=\frac{5-6q}{5}$, which means that $D_0(p)\geq 0$ when $p$ is between $\displaystyle \frac{3q}{5}$ and $\displaystyle \frac{5-6q}{5}$ and $D_1(p)\geq 0$ otherwise.

The final task is to check that $p=\frac{3}{5}q$ and $p=\frac{5-6q}{5}$ are indeed zeros of $D_0$ and $D_1$. For $p=\frac{3}{5}q$ we have
    $$\mathcal{N}_{\frac{3q}{5},q}(|\xi\ra\la \xi|)= \frac{5-3q}{5}\cdot |\xi\ra\la \xi|+\frac{3q}{5}\cdot \text{Id}_3,$$
whose eigenvalues are the same for any choice of $\xi$. Thus we obtain $D_0(p)=0=D_1(p)$. The other case is the same from the formula
    $$\mathcal{N}_{\frac{5-6q}{5},q}(|\xi\ra\la \xi|)= \frac{9q-5}{10}\cdot R_2|\overline{\xi}\ra\la \overline{\xi}|R_2^*+\frac{15-9q}{10}\cdot \text{Id}_3,$$
where $R_2$ is from \eqref{eq-R_m}. Our conclusion can be visualized as follows particularly for $p,q\geq 0$ with $p+q\leq 1$:

\begin{figure}[hbt!]
    \centering
		\includegraphics[scale=0.35]{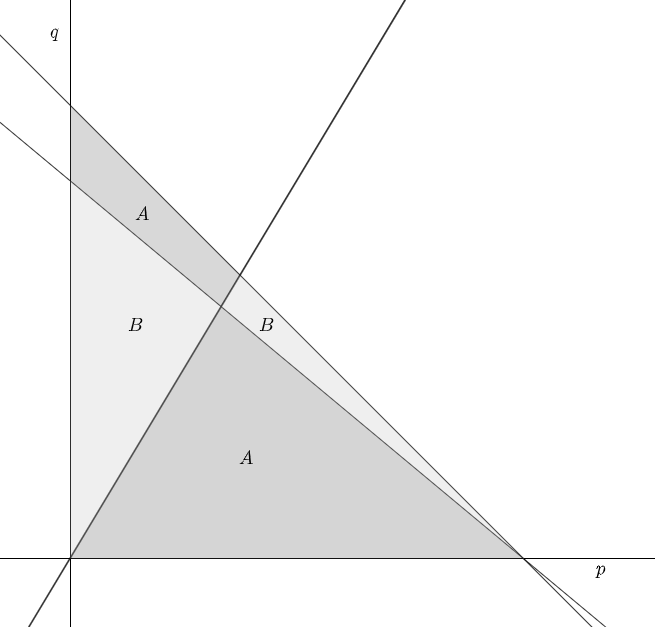}
\caption{ Minimizers of MOE in \text{CovQC}(2,2)}
\label{MOE(2,2)}
\end{figure}

Here, we have $D_0(p)\geq 0$ in the region $A$ and $D_1(p)\geq 0$ in the region $B$.
\end{proof}

\section{Almost superactivation of coherent information of $SU(2)$-covariant channels}

As discussed in Section \ref{sec-degradability}, degradability is extremely rare in $\text{CovQC}(1,l)$ and $\text{CovQC}(2,2)$, which implies that additivity question of coherent information is much more complicated in this class. Indeed, such a question for $\text{CovQC}(1,1)$ is exactly same with the case of qubit depolarizing channels and has been studied in \cite{DiShSm98,SmSm07,FeWh08}. In particular, for the depolarizing channel 
\begin{equation}
\Phi(\rho)=\left (1-\frac{4p}{3}\right )\rho+\frac{4p}{3}\cdot \frac{1}{2}\text{Id}_2
\end{equation} 
with $0.1893\leq p\leq 0.19125$, we have a stronger form of non-additivity of coherent information, namely {\it superactivation} of coherent information for any $n\geq 3$ \cite{DiShSm98,SmSm07,FeWh08}: $\displaystyle \frac{Q^{(1)}(\Phi^{\otimes n})}{n}>Q^{(1)}(\Phi)=0$. Also, note that a recent paper \cite{LLS18} revealed the superadditivity of the coherent information for a class of dephrasure channels $\Psi$,
\begin{equation}
    \frac{Q^{(1)}(\Psi\otimes \Psi)}{2}>Q^{(1)}(\Psi)=0.
\end{equation}

As such, we might expect a similar phenomenon as above among a general element $\Phi_p\in \text{CovQC}(1,l)$ from \eqref{eq-Cov(1,l)-ch}. Indeed, for $l=2$ and $p=0.1045$ we do have
    \[\frac{Q^{(1)}(\Phi_p\otimes \Phi_p)}{2}\geq 0.0039>Q^{(1)}(\Phi_p)\approx 0,\]
which we call an ``{\it almost superactivation} of coherent information". 

%Computing explicit eigenvalues of $\Phi$ for each $p$, superactivation phenomena can be exhibited as follows numerically:
%\begin{itemize}
%\item $p=0.103$ $\Rightarrow$ $Q^{(1)}(\Phi\otimes \Phi)\geq 0.01293 >5\cdot Q^{(1)}(\Phi)\approx 0.01223$
%\item $p=0.1045$ $\Rightarrow$ $Q^{(1)}(\Phi\otimes \Phi)\geq 0.0068$ and $Q^{(1)}(\Phi)\approx 0$.
%\end{itemize}
%Note that both of the above results are much stronger than the usual superactivation $Q^{(1)}(\Phi\otimes \Phi)>2\cdot Q^{(1)}(\Phi)$.

Let us elaborate how we get the above. Recall that 
    $$Q^{(1)}(\Phi) = \max_{\rho\in \mathcal{D}(\Comp^2)}\left \{H(\Phi(\rho))-H(\Phi^c(\rho)) \right\}.$$
We already know $\Phi_p$ is $SU(2)_{(1,l)}$-covariant and the Kraus operators of $\Phi_p$ are given by 
\[\sqrt{1-p}K_1,\cdots,\sqrt{1-p}K_l,\sqrt{p}K_{l+1},\cdots,\sqrt{p}K_{2l+2}\]
described in Subsection \ref{subsubsec-(1,l)}. Then $\Phi_p^c(\rho)$ can be written as $(\text{Tr}\otimes \text{id}_{2l+2})(v \rho v^*)$ where $v:H_1 \hookrightarrow H_{l}\otimes (H_{l-1}\oplus H_{l+1})\cong (H_l\otimes H_{l-1})\oplus (H_l\otimes H_{l+1})$ is given by
\[v(\xi) = \sqrt{1-p}\cdot  \alpha^{l,l-1}_1(\xi)\oplus \sqrt{p}\cdot  \alpha^{l,l+1}_1(\xi),~\xi\in H_1.\]
This isometry $v$ satisfies the following covariance property
\[v(\pi_1(U)\xi)=\left (\pi_l(U)\otimes [\pi_{l-1}(U)\oplus \pi_{l+1}(U)]\right )v(\xi),~U\in SU(2),\]
implying
%its Stinespring isometry is given by $\sqrt{1-p}\cdot \alpha^{l,l-1}_1\oplus \sqrt{p}\cdot \alpha^{l,l+1}_1:H_1 \hookrightarrow H_{l}\otimes (H_{l-1}\oplus H_{l+1})$, so its complementary channel $\Phi^c$ \textcolor{red}{\bf (explain that this model is different from the one using Kraus operators and coherent info. is indepedent of the choice of $\Phi^c$-model)} satisfies covariance with respect to $(\pi_1, \pi_{l-1} \oplus \pi_{l+1})$, i.e. 
    $$\Phi^c(UDU^*)=  \left [ \begin{array}{cc} \pi_{l-1}(U)& 0 \\ 0 &\pi_{l+1}(U) \end{array}\right ] \Phi^c(D) \left [ \begin{array}{cc} \pi_{l-1}(U)^*& 0 \\ 0 &\pi_{l+1}(U)^* \end{array}\right ]$$
for any $U\in SU(2)$ and $D\in M_2(\Comp)$. Applying diagonalization and the covariance properties of $\Phi_p$ and $\Phi_p^c$ we get
    $$Q^{(1)}(\Phi_p) = \max_{D\in \mathcal{D}(\Comp^2):\text{ diagonal}}\left \{H(\Phi_p(D))-H(\Phi_p^c(D)) \right\}.$$
Now let $D_{\lambda}=\left[\begin{array}{cc} \lambda & 0 \\ 0 & 1-\lambda \end{array}\right]$ then we can write down all the eigenvalues of $\Phi(D_{\lambda})$ from \eqref{eq-l-1} and \eqref{eq-l+1}. Moreover,
\eqref{eq-Kraus1} and \eqref{eq-Kraus2} allows us to write down all the eigenvalues of $\Phi^c(D_{\lambda})$ as well, so we have a closed form formula of $H(\Phi_p(D_{\lambda}))-H(\Phi_p^c(D_{\lambda}))$.

In particular, focusing on a specific case $l=2$ with $p=0.1045$, we have a numerical approximation
\begin{equation}
    \max_{0\leq j\leq 10^8}\left \{ H(\Phi_p(D_{\lambda}))-H(\Phi_p^c(D_{\lambda})): \lambda=j/10^8\right\}=0. 
\end{equation} 
Then the Fannes-Audenaert inequality and elementary triangle inequalities allow us to have an error estimate $\epsilon\leq  10^{-6}$, so we should have $Q^{(1)}(\Phi_p)\leq  10^{-6}$ up to very small errors in computer calculations.

%An interesting feature is that all the non-trivial coherent information has been attained at the maximally mixed state $\displaystyle \frac{1}{2}\cdot \text{Id}_2$. Indeed, for fixed $l$, the graphs of $H(\Phi(D))-H(\Phi^c(D))$ of $D=\left [ \begin{array}{cc}\lambda&0\\0&1-\lambda \end{array} \right ]$ are given as follows.

%\begin{center}
	%\includegraphics[scale=0.5]{figure6}
%\end{center}

On the other hand, using the mixed state $\rho=\displaystyle \frac{1}{2}|00\ra\la 00| +\frac{1}{2}|11\ra\la 11|$ for the case for $l=2$ with $p=0.1045$, we get a lower bound of $Q^{(1)}(\Phi_p\otimes \Phi_p)$ as follows:
\begin{align*}
\frac{Q^{(1)}(\Phi_p\otimes \Phi_p)}{2} & \ge \frac{1}{2}\cdot \left ( H((\Phi_p\otimes \Phi_p)(\rho))-H((\Phi_p^c\otimes \Phi_p^c)(\rho))\right )\\
&\approx \frac{1}{2}\cdot (2.0727 - 2.0648) \approx 0.0039.
\end{align*}
Thus we obtain the following almost superactivation for $\Phi_p$:
    $$\frac{Q^{(1)}(\Phi_p\otimes \Phi_p)}{2}\geq 0.0039 >10^{-6}\geq Q^{(1)}(\Phi_p).$$

\begin{remark}
\begin{enumerate}
\item The exact value of quantum capacity $Q(\Phi_p)$ is still open even for $l=1$ which corresponds to the case of qubit depolarizing channels.
\item A general form of $\rho=\displaystyle \frac{1}{2}|00\ra\la 00| +\frac{1}{2}|11\ra\la 11|$ in $n$-qubit systems is used in \cite{LLS18} to prove superadditivity of the coherent information of dephrasure channels.
\end{enumerate}
\end{remark}

\emph{Acknowledgements}: H.H. Lee was supported by the Basic Science Research Program through the National Research Foundation of Korea (NRF) Grant NRF-2017R1E1A1A03070510 and the National Research Foundation of Korea (NRF) Grant funded by the Korean Government (MSIT) (Grant No.2017R1A5A1015626). S-G. Youn was supported by the New Faculty Startup Fund from Seoul National University. E. Chang, J. Kim, H. Kwak and S-G.Youn were supported by the National Research Foundation of Korea (NRF) grant funded by the Korea government (MSIT) (No. 2020R1C1C1A01009681).

\appendix

\section{Positivity and decomposability of irreducibly $SU(2)$-covariant linear maps}

While complete positivity of linear maps ensures physical nature of the maps, positivity of linear maps is also important in QIT since positive non-CP maps can be used as entanglement witnesses. In this appendix we focus on positivity of $SU(2)_{(k,l)}$-covariant linear maps. Recall that any $SU(2)_{(k,l)}$-covariant linear map is of the form $\Phi = \displaystyle \sum_{r=0}^{\min\left\{k,l\right\}} a_r \Phi^{k\rightarrow l}_{|k-l|+2r}$, $a_r \in \Comp$ and complete positivity of $\Phi$ is equivalent to the condition that all the coefficients $a_r$ are non-negative. However, positivity of $\Phi$ is not immediate from this decomposition, and such a question of characterizing positive irreducibly covariant maps was raised as an open problem in \cite{MSD17}. Subsequently, positive irreducibly covariant linear maps for the permutation group $S(3)$ and the quaternion group $Q$ have been characterized in \cite{KMS20}.

In this section, we exhibit all positive linear maps with $SU(2)_{(k,l)}$-covariance for $(k,l) = (1,l)$, $(k,1)$ and $(2,2)$, and prove that all such positive maps are automatically {\it decomposable}, i.e. sums of completely positive maps and completely co-positive maps, which is analogous to \cite[Theorem 13]{KMS20}. Let us begin with the case of $SU(2)_{(1,l)}$-covariance.

\begin{proposition}\label{prop-pos1}
	Let $\Phi_p = (1-p)\Phi^{1\rightarrow l}_{l-1}+p\Phi^{1\rightarrow l}_{l+1}$ for $p\in \Comp$ and $l\in \n$. Then we have
	    $$\text{$\Phi_p$ is positive} \Leftrightarrow \displaystyle 0\leq p\leq \frac{l+2}{l+1} \Leftrightarrow \text{$\Phi_p$ is decomposable}.$$
\end{proposition}
\begin{proof}

Note that $p=0$ gives us a quantum channel, which is completely positive. We first check that the other end point $p=\displaystyle \frac{l+2}{l+1}$ gives us a completely co-positive map $\displaystyle \Phi_p =\left (1-\frac{l+2}{l+1}\right )\Phi^{1\rightarrow l}_{l-1}+\frac{l+2}{l+1}\Phi^{1\rightarrow l}_{l+1}$. Indeed, we can easily see that $\Phi_p(\rho)=R_l \left ( \Phi^{1\rightarrow l}_{l-1}(\rho)\right )^t R_l^*$, where $R_l$ is a unitary given in \eqref{eq-R_m}.

For the converse direction we observe
\begin{align*}
\Phi_p(|0\ra\la 0|)& = \frac{2}{l(l+1)}\sum_{i=0}^l \left ( (1-p)\cdot (l-i)+p\cdot \frac{l(i+1)}{l+2}\right ) |i\ra\la i|\\
\Phi_p(|1\ra\la 1|)&= \frac{2}{l(l+1)}\sum_{i=0}^l \left ((1-p)\cdot i+p\cdot \frac{l(l-i+1)}{l+2} \right ) |i\ra\la i|,
\end{align*}
which gives us
\begin{align*}
\la 0|\Phi(|0\ra\la 0|)|0\ra&=\frac{2}{l+1}\left (1-p\cdot \frac{l+1}{l+2} \right ) \geq 0\\
\la 0| \Phi(|1\ra\la 1|)|0\ra&=p\cdot \displaystyle \frac{2}{l+2}\geq 0.
\end{align*}
This leads us to the desired conclusion.
\end{proof}

By duality we immediately get the following.

\begin{proposition}
	Let $\Psi_p = (1-p)\Phi^{l\rightarrow 1}_{l-1}+p\Phi^{l\rightarrow 1}_{l+1}$ for $p\in \Comp$ and $l\in \n$. Then we have
	    $$\text{$\Psi_p$ is positive} \Leftrightarrow \displaystyle 0\leq p\leq \frac{l+2}{l+1} \Leftrightarrow \text{$\Psi_p$ is decomposable}.$$
\end{proposition}

Finally, we focus on the set of all positive $SU(2)_{(2,2)}$-covariant maps.

\begin{theorem}
	Let $\mc N_{p,q}= (1-p-q)\Phi^{2\rightarrow 2}_0+p\Phi^{2\rightarrow 2}_2+q\Phi^{2\rightarrow 2}_4$ for $p,q\in \Comp$. Then we have
	    $$\text{$\mc N_{p,q}$ is positive} \Leftrightarrow \begin{cases}q\geq 0\\ 0\leq 5p+3q\leq 5\\ 5p+9q\leq 10\end{cases} \Leftrightarrow \text{$\mc N_{p,q}$ is decomposable}.$$
\end{theorem}
\begin{proof}
As before we focus on the vertices of the trapezoidal region, where two of them  ($\Phi^{2\rightarrow 2}_0$ and $\Phi^{2\rightarrow 2}_2$) are already completely positive. The remaining two vertices are $\displaystyle \Psi := \frac{1}{3}\Phi_{0}^{2 \rightarrow 2}-\Phi_{2}^{2 \rightarrow 2}+\frac{5}{3}\Phi_{4}^{2 \rightarrow 2}$ and $\displaystyle -\frac{1}{2}\Phi^{2\rightarrow 2}_0+\Phi^{2\rightarrow 2}_2+\frac{1}{2}\Psi$ as plotted below.

\begin{center}
\begin{figure}[hbt!]
			\includegraphics[scale=0.3]{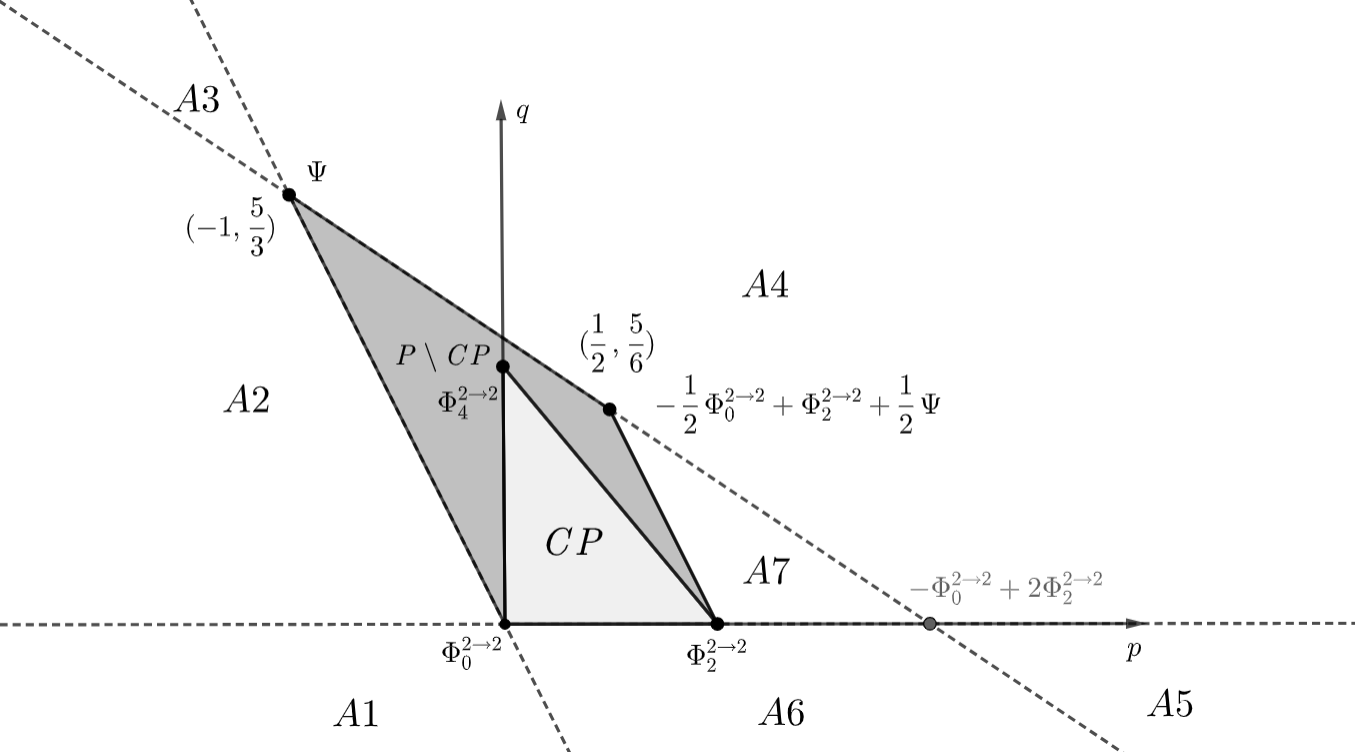}
\caption{ $SU(2)_{(2,2)}$-covariant positive maps}
\end{figure}
\end{center}

Indeed, $\Psi$ and $\displaystyle -\frac{1}{2}\Phi^{2\rightarrow 2}_0+\Phi^{2\rightarrow 2}_0+\frac{1}{2}\Psi$ are completely co-positive since $\Psi(\rho)=R_2\cdot \rho^t \cdot R_2^*$ and
\begin{align*}
\left (-\frac{1}{2}\Phi^{2\rightarrow 2}_0+\Phi^{2\rightarrow 2}_2+\frac{1}{2}\Psi \right )(\rho)&=\frac{1}{2}\left (\text{Tr}(\rho)\text{Id}_3-\rho \right )=R_2 \left (\Phi^{2\rightarrow 2}_2(\rho)\right )^t R_2^*.
\end{align*}
where $R_2=\left [ \begin{array}{ccc} 0&0&1\\ 0&-1&0\\ 1&0&0 \end{array} \right ]$. Thus, all the elements in the trapezoid are decomposable, and consequently positive. 

Lastly, there is no positive map in the regions A1-A7. As in the proof of Proposition \ref{prop-pos1}, canonical diagonal density matrices $|i\ra\la i|$ are enough to get the conclusion. Indeed, A1-A6 are excluded from positivity of $\Phi(|0\ra\la 0|)$ and $\Phi(|2\ra\la 2|)$, and A7 is also excluded due to positivity of $\Phi(|1\ra\la 1|)$.

\end{proof}

\bibliographystyle{alpha}
\bibliography{Reference}

\end{document}